\documentclass[lettersize,journal]{IEEEtran}
\usepackage[utf8]{inputenc}
\usepackage[cmex10]{amsmath}
\usepackage{amsthm}
\usepackage{amsbsy,pifont,verbatim,amsfonts,amssymb}
\usepackage{array}
\usepackage[caption=false,font=normalsize,labelfont=sf,textfont=sf]{subfig}
\usepackage{textcomp}
\usepackage{stfloats}
\usepackage{url}
\usepackage{verbatim}
\usepackage{algorithm2e}
\usepackage{hyperref}
\usepackage{xcolor}
\usepackage{nomencl}
\usepackage{etoolbox}
\usepackage{graphicx}
\usepackage[style=ieee, backend=biber, doi=false, maxnames=2, minnames=1, url=false]{biblatex}
\newtheorem{theorem}{Theorem}
\newtheorem{lemma}[theorem]{Lemma}

\begin{document}

\title{Continuous Switch Model and Heuristics for Mixed-Integer Problems in Power Systems}

\author{
    \IEEEauthorblockN{Aayushya Agarwal\IEEEauthorrefmark{1}, Amritanshu Pandey\IEEEauthorrefmark{2},Larry Pileggi\IEEEauthorrefmark{1}}

\thanks{A. A., L. P. are with the Department of Electrical and Computer Engineering, Carnegie Mellon University, Pittsburgh, PA, 15213 USA (email:\{aayushya,pileggi\}@andrew.cmu.edu), A.P. is with the Electrical and Biomedical Engineering Department in University of Vermont (email:amritanshu.pandey@uvm.edu)}
}

\markboth{Journal of \LaTeX\ Class Files,~Vol.~14, No.~8, August~2021}%
{Shell \MakeLowercase{\textit{et al.}}: A Sample Article Using IEEEtran.cls for IEEE Journals}


\maketitle

\begin{abstract}
Many power systems operation and planning computations (e.g., transmission and generation switching and placement) solve a mixed-integer nonlinear problem (MINLP) with binary variables representing the decision to connect devices to the grid. Binary variables with nonlinear AC network constraints make this problem NP-hard. For large real-world networks, obtaining an AC feasible optimum solution for these problems is computationally challenging and often unattainable with state-of-the-art tools today. In this work, we map the MINLP decision problem into a set of equivalent circuits by representing binary variables with a circuit-based continuous switch model. We characterize the continuous switch model by a controlled nonlinear impedance that more closely mimics the physical behavior of a real-world switch. This mapping  effectively transforms the MINLP problem into an NLP problem. We mathematically show that this transformation is a tight relaxation of the MINLP problem. For fast and robust convergence, we develop physics-driven homotopy and Newton-Raphson damping methods. To validate this approach, we empirically show robust convergences for large, realistic systems ($>$ 70,000 buses) in a practical wall-clock time to an AC-feasible optimum. We compare our results and show improvement over industry-standard tools and other binary relaxation methods. 
\end{abstract}


\section{Introduction}
\label{sec:intro}
\IEEEPARstart{O}{peration} and planning studies in power systems rely on grid optimizations to provide timely decisions to improve the reliability and efficiency of grid configurations \cite{majidi2018optimization}. An increasingly important new subset of these optimizations requires \textit{optimally} switching in and out grid equipment while satisfying AC network constraints. The importance of reliably solving this optimization problem is highlighted in a recent renewable expansion study by Midcontinent Independent System Operator (MISO) which needed to transform DC-network constrained production cost model into an AC-feasible power model by adding an optimal number of lines and shunts \cite{bakke2019renewable}. Other studies that optimally switched grid devices for operations include a recent grid-optimization (GO) challenge \cite{challenge} in which discrete shunts, transformers, and lines were optimally switched to ensure \textit{maximal} feasibility of security constraints within the AC optimal power flow (OPF) paradigm. We denote the general class of MINLP (with binary variables) for power systems applications that determine an optimal addition or removal of devices to the grid as an \emph{optimal-decision problem} (ODP) and a focus in this paper.

While methodologies to solve ODPs are gaining attention from industry \cite{challenge} and academia \cite{zhang2012transmission,danna2005exploring,ma2020unit,huang2016optimal,huang2017branch,haghighat2018bilevel,gao2022internally,ke2015modified,lopez2021optimal,chen2016improving,nasri2015network,miao2015capacitor,gao2022internally, alguacil2003transmission,ruiz2015robust,li2022mixed,ghaddar2019power} alike, ODPs that consider nonlinear AC-constraints remain challenging to solve due to two key difficulties: 1) binary variables creating a discontinuous solution and gradient space, and 2) nonconvex solution space due to nonlinear AC constraints. State-of-the-art methods to solve ODPs either i) relax the nonconvex AC-network constraints to create a linear integer relaxation (i.e., MILP formulation) \cite{zhang2012transmission,danna2005exploring,ma2020unit,huang2016optimal,huang2017branch,haghighat2018bilevel,gao2022internally,ke2015modified,lopez2021optimal,chen2016improving,nasri2015network,miao2015capacitor,gao2022internally, alguacil2003transmission,ruiz2015robust,li2022mixed,ghaddar2019power}, or ii) relax the binary variables to create a fully continuous nonconvex solution space (i.e., NLP formulation) \cite{zhang2012transmission,de2005transmission}. The challenge with the former method is that it produces non-operational DC-feasible solutions that do not satisfy AC-feasibility \cite{baker2021solutions} 
Many ODP applications, like studying voltage stability of high renewable penetration scenarios \cite{de2005transmission}, require satisfying AC-network constraints. Currently, in industry practice \cite{challenge}, engineers require significant post-processing to convert DC solutions to AC-feasible solutions \cite{vyakaranam2021automated, bakke2019renewable}. We posit a more practical approach should directly provide \textit{fast} AC-feasible solutions to ODPs.

The latter approach to solve OPDs provides AC-feasible solutions by relaxing the binary variables. This transforms the underlying MINLP problem into an NLP problem with AC-network equations in the constraint set. However, the challenge with state-of-the-art methods \cite{zhang2012transmission,de2005transmission} are the steep nonconvex solution space and nonlinearities due to binary relaxation that prevent gradient-based solvers \cite{byrd2006k,wachter2006implementation} from scaling to large systems. Other heuristics are often necessary in tandem to solve real-world problems \cite{mcnamara2022two}. Most current approaches are slow or do not provide reasonable guarantees for the quality of binary relaxations \cite{mcnamara2022two,de2005transmission}.

We present a novel Grid Integer Switch Model for Optimization (GISMO), which uses a continuous, \emph{circuits-inspired binary relaxation} and designs accompanying heuristics to provide scalability and robustness. The main contributions of GISMO are that i) it guarantees an AC-feasible solution for large-scale ODP problems, ii) provides \textit{resonable} guarantees on the \textit{tightness} of binary relaxations, iii) introduces accompanying heuristics to optimize large-scale networks. With these features , GISMO is an asset for many emerging industry problems \cite{bakke2019renewable} that require solving ODPs to AC-feasible solutions.

Starting with a circuit-based formulation of the underlying grid, GISMO first uses an ideal series switch to represent binary decisions in the ODP. The ideal switch is inserted in series with potential devices to exactly represent the decision to switch in/out of the specific device (similar to the workings of a circuit-breaker). For example, an ideal switch can be placed in series with a generator to represent decisions in generation expansion problems. We then define a physics-inspired, continuous relaxation of the switch to provide a fully continuous and operable solution space. %The continuous relaxation uses a diode-like exponential characteristic to transform the original integrality constraint in the MINLP to a relaxed indicator function. 
The binary relaxation in the continuous switch model is motivated by the physical behavior of a real-world switch (e.g., diode, transistors), which has minimal power drop in the "on" state. With this insight, we design the relaxation limits based on the maximum power drop in the continuous switches, thereby providing sufficient guarantees of the tightness bounds for the binary relaxations.

Solving the relaxed NLP remains difficult; however, the physics-inspired relaxation has significant benefits over other binary relaxations \cite{de2005transmission,zhang2012transmission}. It enables the development of problem-specific heuristics such as scalable and robust homotopy and Newton-Raphson (NR) damping methods which, in the past, have helped robustly solve similar circuits with billions of similar switch-like models \cite{pillage1998electronic}. Unlike neighborhood search methods \cite{danna2005exploring,ma2020unit}, these heuristics utilize domain-knowledge from the underlying switch model to provide robust convergence by ensuring a physically realistic behavior throughout the solution path. The novelty of the proposed approach is in the development of:
\begin{enumerate}
    \item A novel binary relaxation inspired by circuit formalism to transform power systems MINLP to NLP with bounds on the tightness of relaxations
    \item Novel physics-based homotopy and NR dampening methods for robust  and scalable convergence of large real-world networks with such continuous switch models
    \item Guarantees on obtaining AC-feasible solutions for ODPs
\end{enumerate}

We demonstrate the generalizability and robustness of GISMO by applying it to numerous operation and planning analyses. We compare the efficacy of GISMO approach against several industry-standard mixed-integer solvers and NLP solvers for large-realistic networks such as the Eastern Interconnection with over 70,000 buses.

\section{Prior Work}
\label{sec:prior_work}

Methodologies to solve ODPs in power systems generally use one of two approaches: relaxing the network constraints or relaxing the binary variables.
\subsection{Relaxing the Nonlinearities}
These approaches for solving ODPs linearize the nonlinear constraints to form a MILP that is solved using algorithms like branch-and-bound methods \cite{huang2016optimal, haghighat2018bilevel, ke2015modified, lopez2021optimal}, branch-and-cut \cite{huang2017branch,gao2022internally,sousa_branch_cut}, or genetic algorithm approaches \cite{salkuti2018congestion, hartono2019optimal, mahdavi2020genetic}. Branch-and-bound methods have often been used in expansion analysis \cite{huang2017branch,gao2022internally, alguacil2003transmission,ruiz2015robust,li2022mixed,ghaddar2019power}, unit commitment (UC) \cite{ke2015modified} and optimal shunt placement \cite{lopez2021optimal}. To make the problem tractable, previous works rely on Bender’s decomposition \cite{nasri2015network,miao2015capacitor,huang2017branch}, Lagrange relaxation \cite{chen2016improving} and neighborhood search methods \cite{miao2015capacitor,huang2017branch} for an efficient branch-cutting methodology. However, by linearizing the non-linear AC constraints, the solutions are not AC feasible \cite{baker2021solutions}. Scalability is also problematic with MILP, especially when many real-world expansion problems include millions of variables \cite{xu2009computational}.

Genetic algorithms have also been used for transmission switching problems \cite{salkuti2018congestion}, optimal shunt placement \cite{hartono2019optimal} and expansion analysis \cite{mahdavi2020genetic}. However, genetic algorithms generally require many samples to provide an optimal solution reliably. As a result, previous works have limited scalability as generating and solving sufficient samples is intractable.

%\subsection{Relaxing Nonlinear Constraints}
%In another MINLP solution approach for power systems,  previous works have relaxed the nonlinear AC network constraints using a DC approximation \cite{zhang2012transmission,danna2005exploring,ma2020unit,huang2016optimal,huang2017branch,haghighat2018bilevel,gao2022internally,ke2015modified,lopez2021optimal,chen2016improving,nasri2015network,miao2015capacitor,gao2022internally, alguacil2003transmission,ruiz2015robust,li2022mixed,ghaddar2019power}. By removing the nonlinear equality constraints, these works essentially transform the MINLP into a MILP (mixed-integer linear problem), with guaranteed convergence. Commercial MILP solvers such as Gurobi \cite{gurobi2021gurobi} are used to solve the underlying MILP. However, solving the problem using linearized constraints does not satisfy grid physics \cite{baker2021solutions} and subsequent postprocessing is always necessary. In Section VI, we compare our methodology, which satisfies realistic nonlinear constraints, to the linearized constrained problem and demonstrate how the linearized constraints cannot trivially map to a feasible AC solution.

\subsection{Relaxing Integrality Constraints}
Another promising but less explored approach is relaxing the integrality constraints rather than the nonlinear ones. For instance, prior works have relaxed the integrality constraints using continuous functions such as a penalized quadratic representation $(x(x-1))$ \cite{zhang2012transmission} or a sigmoid function \cite{de2005transmission}. These methods transform the MINLP into an NLP that remains difficult to solve due to the added nonlinearities. Nonlinear relaxations of binary variables using sigmoids \cite{de2005transmission} are extremely steep and challenging to deal with using gradient-based methods. Existing methods have not demonstrated the ability to scale and solve large realistic test cases using this approach. 
Moreover, these methods do not guarantee the goodness of the solution, as the optimality gap is hard to estimate. More recently, an approach in \cite{mcnamara2022two} used a circuits-based continuous relaxation for optimizing discrete device controls by relaxing the discrete device model and was able to optimize large-sized systems. However, the heuristics and models were developed for discrete control elements specifically and did not apply to general mixed-integer problems.
 
In this paper, we also relax the integrality constraints using a circuit-based model. However, unlike prior work, we use real-world device physics to develop the switch model to embed the binary decision variables into the optimization as continuous variables. This method further serves as a basis for developing scalable homotopy and NR-dampening methods to enforce robust convergence. These methods allow us to address the gaps of other relaxed algorithms by 1) ensuring AC network constraints are satisfied and 2) developing model-specific heuristics to scale the methodology to solve Eastern-Interconnection sized problems.

\section{Defining Optimal Decision Problem}

ODPs in power systems determine an optimal network configuration by adding or removing devices to the grid to meet certain objectives. For example, a subset of ODP problems requires adding new generation to the network. This application tackles the binary switching aspect of a larger unit commitment problem. Similarly, optimal transmission switching connects a transmission line to the network to reduce congestion or improve network feasibility. These decisions can be mathematically represented by a vector of binary variables, $X_b=\{x_b|x_b \in \{0,1\}\}$, where a value of 0 disconnects the device and 1 connects the device.

The ODP can be represented as the following MINLP: 
\begin{subequations}
\label{eq:odp}
\begin{align}
    \min_{X_c, X_b} f(X_c, X_b) \:\:\:
    s.t. \label{odp_obj}\\
    g_f(X_c)+g_d(X_c,X_b)=0 \label{eq:odp_kcl}\\
    h_f(X_c)+h_d(X_cX_b)\leq 0 \label{eq:odp_ineq}\\
    X_b \in \{0,1\}
    \label{eq:odp_integrality}
\end{align}
\end{subequations}

\noindent where $X_c$ is a state vector of continuous variables that include real and imaginary voltages($V_r$ and $V_i$ respectively), active power generation ($P_G)$ and reactive power generation ($Q_G$). The system is constrained to satisfy nonlinear AC network constraints (derived from Kirchhoff's current laws), represented by \eqref{eq:odp_kcl}. The set of network constraints relating to the fixed devices in the grid is represented by the nonlinear equations, $g_f(X_c)$. In contrast, the network constraints relating to the devices that can be added or removed from the grid are represented by $g_d(X_c,X_b)$. For example, in applications that require adding new generation to the network, $g_d(X_c,X_b)$ captures the power or current contribution for the additional generators. $g_d(X_c,X_b)$ is a function of continuous variables$X_c$ and binary variables $X_b$, with following behavior:
\begin{equation}
    g_d = \begin{cases}
    g_d(X_c,1), & X_b=1 \\
    0, & X_b=0
    \end{cases}
\end{equation}
which disconnects the device from the network when $X_b=0$.

The device operating limits are represented through the inequality constraints in \eqref{eq:odp_ineq}, which are split into the fixed device constraints, $h_f(X_c)$ and devices that can be added or removed, $h_d(X_c,X_b)$.

The objective function, $f(X_c,X_b)$, in \eqref{eq:odp}, is designed to model the objective of the ODP application as a function of $X_c$ and $X_b$. For example, $f(X_c,X_b)$ can represent the operating and startup cost of commissioning an optimal set of participating generators,. Similarly, $f(X_c,X_b)$ can represent the cost of switching a transmission line in optimal transmission line switching problems where the binary state vector models the transmission line switch positions. We generalize the objective function to include a startup and/or switching cost, $f_{su}(X_b)$ as well as an operational cost ($f_{op}(X_c,X_b)$), which are functions of $X_b$ and $X_c$ as shown below:
\begin{equation}
    f(X_c,X_b)=f_{su}(X_b)+f_{op}(X_c,X_b)
    \label{eq:total_obj}
\end{equation}

\section{Continuous Switch Model for ODP} \label{sec:switch_model}

\subsection{Modeling Binary Variables with Nonlinear Conductance}

Unlike general MINLP problems, ODPs are unique because the binary variables are associated with the physical action of connecting/disconnecting devices to the grid. 
This switching behavior is physically implemented using switches such as relays or circuit breakers placed in series between a device and the grid. A more realistic model of the binary actions in the ODPs is to use a model of the physical switch. Inspired by circuit simulation, we model this switch behavior by conductance $G^{id}_{sw}\in\{0,G_{closed}\}$ that is placed in series between a device and the grid, as shown in Fig. \ref{fig:continuous_switch}. The conductance, $G^{id}_{sw}$ is controlled to have a value of either 0, which electrically disconnects the device from the grid, or a very large value of $G_{closed}$, which connects the device to the grid with minimal loss. This promotes a physically-inspired model of the binary variables in the ODP \eqref{eq:odp}, which becomes a basis for this work's continuous relaxation and heuristics.

\begin{figure}
    \centering
    \includegraphics[width=\columnwidth]{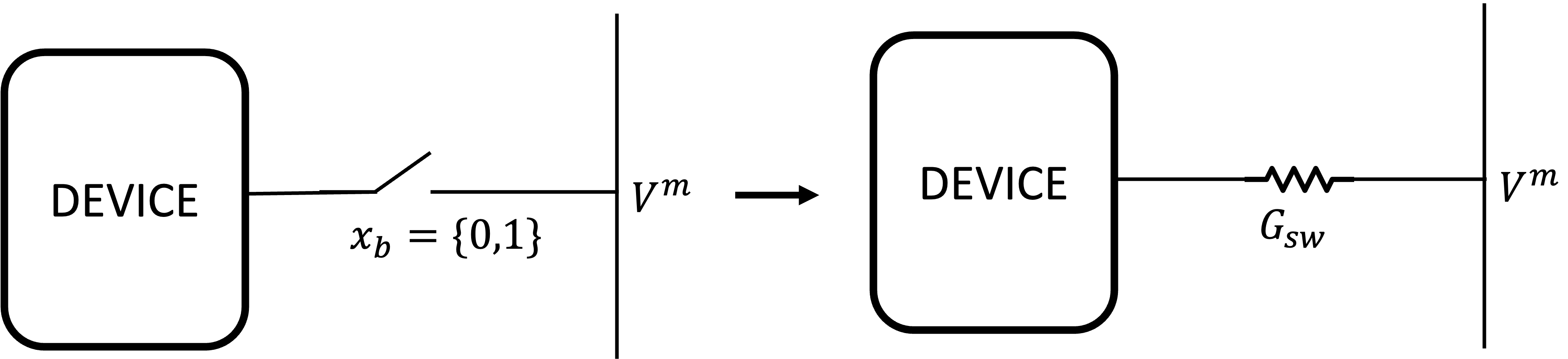}
    \caption{An ideal switch is placed in series between the device and bus (with voltage $V^m$) to embed the binary decision ($x_b$) in ODP. The switch Conductance, $G_{sw}$ replaces ideal switch to mimic binary decision.}
    \label{fig:continuous_switch}
\end{figure}

\iffalse
Without loss of generality, we represent the binary decision of \eqref{eq:odp} with an ideal switch placed in series with a device considered for insertion/disconnection. The binary decision, $x_b=\{0,1\}$, represents the position of the ideal switch. With a value of $x_b=0$, the switch is open and the device is disconnected. For a value of $x_b=1$, the switch is closed and the device is connected to the grid. This replicates the behavior of the binary variables in \eqref{eq:odp} and follows real-life operations, as many operational decisions are communicated to switches that are in series with devices.
\begin{figure}
    \centering
    \includegraphics{Figures/odp_ideal_switch.png}
    \caption{An ideal switch is placed in series between the device and bus (with voltage $V^m$) to embed the binary decision ($x_b$) in ODP.}
    \label{fig:odp_ideal_switch}
\end{figure}
\fi

A unique aspect of the ODP is that each integer variable value of $G^{id}_{sw}$ has a physical controlling mechanism that drives the decision to either connect or disconnect the device to the grid. This physical controlling mechanism, which we denote $V_g$, differs for each power systems application but is a function of the continuous variables of the grid. For example, the driving mechanism for commissioning additional generation to the grid is its active power, $P_G$, which must be greater than the minimum active power generation, ${\underline{P}_G}$. Similarly, the driving mechanism in optimal transmission line switching is the current magnitude, $|I_{tx}|$, where any nonzero value indicates that the transmission line must be connected. This insight allows us to define a  mechanism to control the switch conductance, $G^{id}_{sw}$ between high and low conductances by: 

\begin{equation}
\label{eq:switch_conductance}    G^{id}_{sw} = \begin{cases}
    0, & V_g < V_{th} \\
    G_{closed}, & V_g \geq V_{th},
    \end{cases}
\end{equation}
where $V_g$ is the driving mechanism and $V_{th}$ is a threshold (for example  ${\underline{P}_G}$) that allows the binary variable to switch. The behavior \eqref{eq:switch_conductance} acts as a signal to determine the value of the switch conductance on the right of Fig. \ref{fig:continuous_switch}. 
\subsection{Continuous Relaxation of Control Mechanism}

The switch conductance model $G^{id}_{sw}$ can conduct and block current; however, this behavior in \eqref{eq:switch_conductance} is discontinuous at $V_g=V_{th}$ and cannot be directly incorporated within a gradient-based NR solver \cite{pillage1998electronic}. We first relax the integrality aspect of $G^{id}_{sw}$ to $G_{sw}\in[0,G_{closed}]$ to allow the switch conductance be controlled within the range of 0 and $G_{closed}$. Inspired by diode models in circuit simulation, we use the following smooth approximation to control $G_{sw}$ to mimic the switch-like behavior of \eqref{eq:switch_conductance}:
\begin{subequations}
\begin{equation}
    I_{Gsw} = \text{log}(1+e^{(V_g - V_{th})})
    \label{eq:Imag_eq}
\end{equation}
\begin{equation}
    I_{Gsw}(G_{sw} - G_{closed}) + \varepsilon = 0
    \label{eq:Igsw_eq}
\end{equation}
\end{subequations}
The continuous approximation in \eqref{eq:Imag_eq} uses a softmax function as a continuous signal to indicate whether  $V_g>V_{th}$. The output of the softmax function, $I_{Gsw}$, is used to signal a change in the value of $G_{sw}$ through a relation defined by a continuous function  \eqref{eq:Igsw_eq}. When the signal $I_{Gsw}>0$, this indicates that $V_g>V_{th}$, and by the relation in \eqref{eq:Igsw_eq}, forces $G_{sw}\rightarrow G_{closed}$, thereby closing the switch. However, when $V_g<V_{th}$, $I_{Gsw}\rightarrow0$ by the relation in \eqref{eq:Imag_eq}, thereby letting $G_{sw}\rightarrow0$, which opens the switch.  The function \eqref{eq:Igsw_eq} mimics a perturbed complementarity slackness constraint that approximates the following indicator function \cite{byrd2000trust}:
\begin{equation}
    I_{Gsw}(G_{sw} -G_{closed}) =0 \;\;\; I_{Gsw}>0, G_{sw}\in[0,G_{closed}].
\end{equation}
\eqref{eq:Igsw_eq} relaxes the discontinuous behavior of the indicator by including a small constant, $\epsilon$, to preserve continuity. To achieve almost ideal conditions, we choose a value of $\varepsilon$ close to 0 (around $10^{-6}$), as shown in Fig. \ref{fig:Igsw_function}.
%This relaxation allows us to remove the discontinuous behavior in \eqref{eq:switch_conductance} and allow $G_{sw}=[0,G_{closed}]$. 

\begin{figure}
    \centering
    \includegraphics[width=0.6\columnwidth]{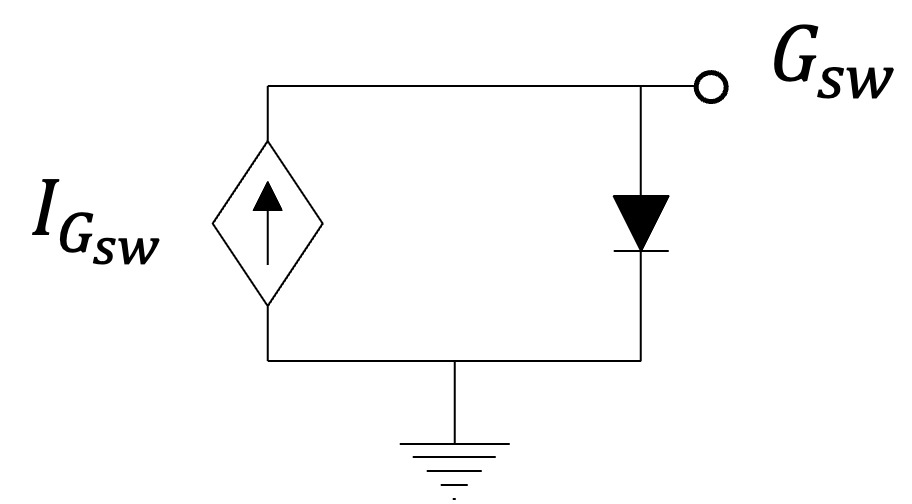}
    \caption{Companion circuit representation of the nonlinear conductance, $G_{sw}$}
    \label{fig:companion_circuit}
\end{figure}

The resulting equation in \eqref{eq:Igsw_eq} can be modeled by an equivalent companion circuit shown in Fig. \ref{fig:companion_circuit}. The companion circuit is a non-physical representation of \eqref{eq:Igsw_eq}, where the circuit's state variable (node voltage) represents the nonlinear conductance value of the switch, $G_{sw}$. The companion circuit includes a current source and a diode-like device. The current source with value $I_{Gsw}$ is a function of the driving mechanism of the switch in \eqref{eq:Imag_eq}. The diode-like device is a circuit representation of equation \eqref{eq:Igsw_eq} (as it resembles an idealized exponential diode function). Table \ref{tab:table_companion_circuit} lists the translation between the variables in \eqref{eq:Igsw_eq} and the state variables in the companion circuit. 
The companion circuit provides a physical analogy of the continuous relation in \eqref{eq:Igsw_eq} that alllows us to directly apply circuit simulation methods, which can robustly simulate circuits with millions of diodes. We derive heuristics using this physical analogy that enable us to scale GISMO to solve large systems.

\begin{table}
\caption{State Variable Translation to Companion Circuit Parameters\label{tab:table_companion_circuit}}
\centering
\begin{tabular}{|c|c|}
\hline
\textbf{Variable} & \textbf{State Variables in Companion circuit} \\
\hline
$G_{sw}$ & Node Voltage \\ 
\hline
$I_{G_{sw}}$ & Current through the diode-like device \\
\hline
%$I_{mag}$ & Controlled Current Source \\
%  \hline 
\end{tabular}
\end{table}

Any nonzero driving mechanism ($I_{Gsw}>0$) will drive the diode-like device in Fig. \ref{fig:companion_circuit} to achieve a voltage drop of $G_{closed}$, thereby forcing $G_{sw}$ to a large value, i.e., the switch is closed. Conversely, when $I_{Gsw}$ is close to 0, the value of $G_{sw}$ ranges between 0 and $G_{closed}$. To push the value of nonlinear conductance to 0 to fully disconnect the device $G_{sw}\rightarrow0$, we minimize the norm of $G_{sw}$ within the start-up objective, $f_{su}$. %however, we include a small valued $\varepsilon$-constant in \eqref{eq:Igsw_eq} to ensure first-order continuity, similar to how perturbed complementary constraints to approximate the indicator function \cite{byrd2000trust}. A smaller value of $\varepsilon$ forces the function \eqref{eq:Igsw_eq} to behave more non-linearly but closer to ideal conditions. 

\begin{figure}
    \centering
    \includegraphics[width=0.6\columnwidth]{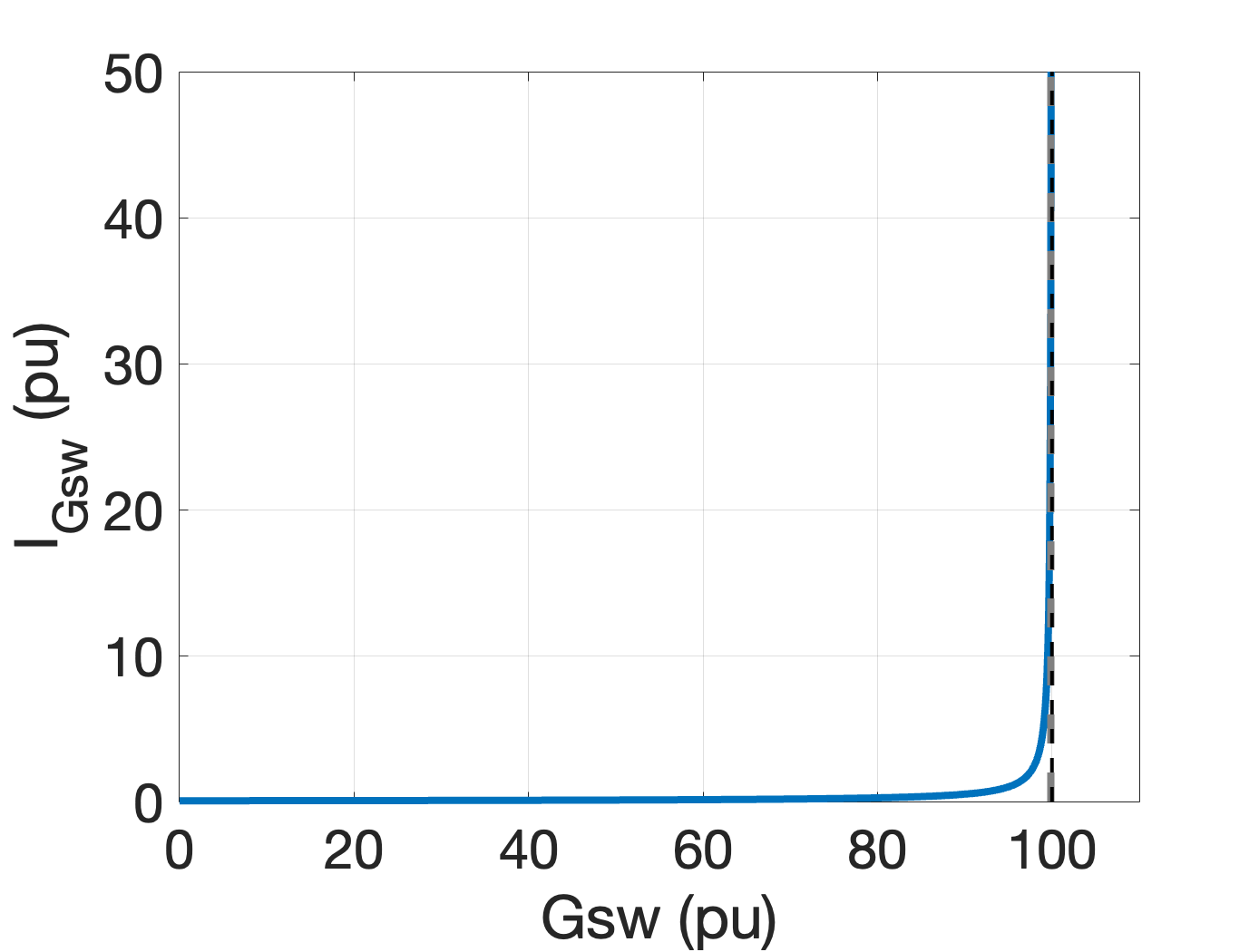}
    \caption{Nonlinear function $I_{Gsw}$ vs $G_{sw}$ with $G_{closed}$ 100pu.}
    \label{fig:Igsw_function}
\end{figure}

%Importantly, the function \eqref{eq:Igsw_eq} and its companion circuit in Fig. \ref{fig:companion_circuit} mimic a perturbed complementary constraint by approximating an indicator function. 
Integrating $G_{sw}$ into the power grid with its companion circuit provides continuous equality and inequality constraints in the original ODP in \eqref{eq:odp}. With $G_{sw}\in[0,G_{closed}]$ taking the role of the binary variable, $x_b={0,1}$, we modify the objective function, $f(X_c,X_b)$ in \eqref{eq:total_obj} to \eqref{eq:total_obj_switch}:

\begin{equation}
    f(X_c, \frac{G_{sw}}{G_{closed}}) = f_{su}(\frac{G_{sw}}{G_{closed}}) + f_{op}(X_c, \frac{G_{sw}}{G_{closed}})
    \label{eq:total_obj_switch}
\end{equation}

\noindent where, division by $G_{closed}$ normalizes the value to remain between 0 and 1 (since $G_{sw}\in[0,G_{closed}]$). %By augmenting the objective with the cost of adding the device, we have defined a multi-objective optimization problem in which the optimal solution is dependent on the proportional weighting of the terms in the objective (for instance, startup cost vs. operational cost). 

%The startup objective, $f_{su}(\frac{G_{sw}}{G_{closed}})$, is necessary for physical operation of companion circuit in Fig. \ref{fig:companion_circuit}. In the absence of a driving mechanism, the value of $G_{sw}$ remains between $[0,G_{closed}$), as represented by the horizontally flat region in Fig. 4. Ideally, we want $G_{sw}=0$ when there is no driving mechanism to represent a disconnected device (i.e., minimize constant costs). As any non-zero $G_{sw}$ has cost associated with it, we additionally minimize the value of $G_{sw}$ in the objective (so that it goes to zero unless required by the system). Like the continuous switch model, the companion circuit modifies the original MINLP ODP in \eqref{eq:odp} into a fully continuous NLP problem, formulated in \eqref{eq:odp_switch}.

\begin{subequations}
\label{eq:odp_switch}
    \begin{equation}
        \label{eq:odp_obj_switch}
        \min_{X_c, G_{sw}} f_{su}(\frac{G_{sw}}{G_{closed}}) + f_{op}(X_c, \frac{G_{sw}}{G_{closed}}) \quad,
        \text{s.t.} 
    \end{equation}
    \begin{equation}
    \label{eq:odp_kcl_switch}
        g_f(X_c) + g_d(X_c, G_{sw})=0
    \end{equation}
    \begin{equation}
    \label{eq:odp_Igsw_switch}
        I_{G_{sw}} = log(1+exp(V_g - V_{th}))
    \end{equation}
    \begin{equation}
    \label{eq:odp_diode_switch}
        I_{Gsw}(G_{sw} - G_{closed}) + \varepsilon=0
    \end{equation}
    \begin{equation}
    \label{eq:odp_ineq_switch}
        h_f(X_c) + h_d(X_c, G_{sw}) \leq 0
    \end{equation}

\end{subequations}

While the relaxed problem in \eqref{eq:odp_switch} provides a fully continuous solution space, it remains NP-hard due to nonlinear equality constraints. We co-design circuit-based heuristics that complement the companion circuit model and $G_{sw}$ switch to ensure robust convergence. 
%In fact, without any heuristics, the current proposed binary relaxation remains very difficult to converge for large systems. Therefore, inspired by our understanding of circuits and best methods to simulate them, we use the switch model in Fig. \ref{fig:continuous_switch} to represent the binary relaxation and develop strong heuristics in the form of homotopy and Newton-Raphson dampening methods to ensure robust convergence of ODPs. 

\subsubsection{Error Bounds on Relaxation}

We provide an upper bound on the binary variable relaxation by analyzing the total power loss in the transmission network. Consider the solution of the network using an ideal switch, with a total power loss in the transmission network denoted by $P_{loss}^{ideal}$. 
In the limit cases where the switches are all closed, GISMO inserts nonlinear conductances into the network, each equal to $\sim G_{closed}$. Each switch conductance dissipates power equal to $P_{diss}=I_{Gsw}^2/{G_{closed}}$.
Then the power loss in the transmission network with the continuous switch model is upper-bounded by:
\begin{equation}
    P_{loss}^{Gsw} \leq P_{loss}^{ideal} + n P_{diss} = P_{loss}^{ideal} + n\frac{I_{Gsw}^2}{G_{closed}}
\end{equation}
where $n$ is the number of switches modeled through the network. This upper bound represents the worst-case loss due to the continuous-switch relaxation.Using a large value of $G_{closed}$, we can bound the binary variable relaxation error within the numerical precision noise of the ideal solution while maintaining a fully continuous solution space. As we increase the value of $G_{closed}$, we move closer to an ideal switch behavior. The upper bound error between the continuous switch relaxation and the ideal switch behavior, given by $P_{loss}^{Gsw} - P_{loss}^{ideal}$ is shown to have an asymptotic functional dependence with $G_{closed}$ in Figure \ref{fig:relaxation_error}

\begin{figure}
    \centering
    \includegraphics[width=0.6\columnwidth]{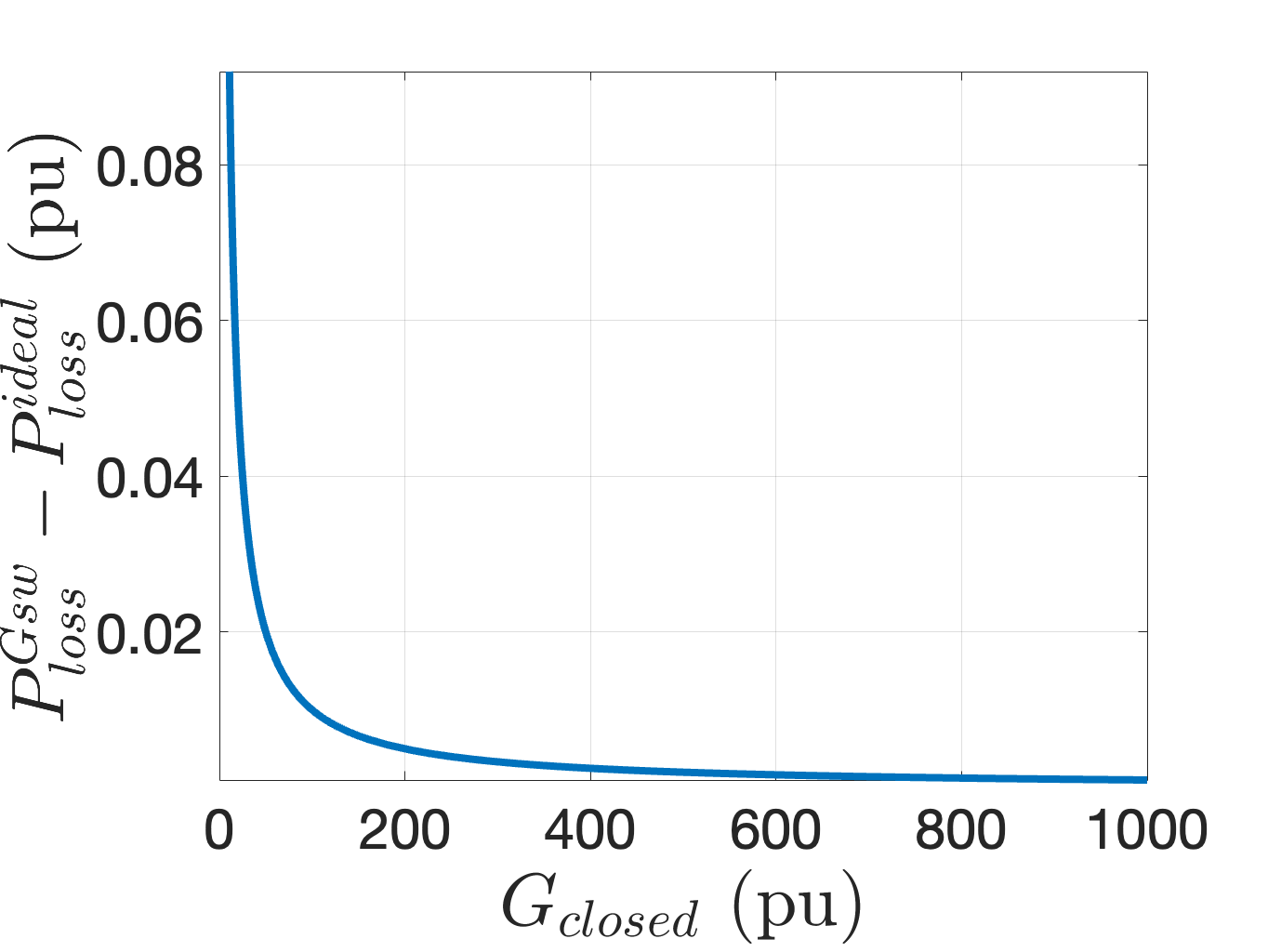}
    \caption{Upper Bound of the Relaxation Error ($P_{loss}^{Gsw} - P_{loss}^{ideal})$ as a function of $G_{closed}$}
    \label{fig:relaxation_error}
\end{figure}
\section{Physics-Driven Heuristics}
\label{sec:heuristics}

The key challenge to problem convergence with the continuous switch model lies in the diode-like device behavior in \eqref{eq:odp_diode_switch}, which has an abrupt change to mimic binary behavior. An abrupt change in the state variable of the companion model ($G_{sw}$) can almost instantly connect a device to the rest of the network and cause an abrupt change in all other grid state variables. This poses a challenge for gradient-based solvers such as Newton-Raphson as it may skip over regions of interest and even diverge.

To avoid this challenge, we utilize key characteristics of the circuit-based switch model. We target two important characteristics of the companion diode circuit in Fig. \ref{fig:companion_circuit}:

\begin{enumerate}
    \item \textit{Reduction in abrupt change to the companion circuit's state variable (voltage), $G_{sw}$, reduces abrupt changes in other state-variables} 
    \item
    \textit{Known steep current-voltage relation of the diode-like device in \eqref{eq:switch_conductance}.}
\end{enumerate}

We leverage these domain-based properties and design two heuristics to ensure robust convergence of the proposed models. The first heuristic combines two physics-inspired \emph{\textbf{homotopy methods}} to avoid abrupt changes to companion circuit and grid network voltages during NR. In the second heuristic, we exploit the known current-voltage relation of the diode-like device to develop a robust \emph{\textbf{NR-dampening algorithm}} that ensures a physically feasible diode-like behavior. 

\subsection{Path-tracing Homotopy Methods}

Homotopy is a class of successive relaxation methods that can handle difficult nonlinear circuit simulation problems. These methods rely on the quadratic convergence of NR by solving a series of sub-problems that trace a path in the solution space, beginning at an easily solvable problem, $E\left(X\right)=0$ and leading to the original problem, $\mathcal{F}\left(X\right)=0$. This is achieved by embedding a scalar homotopy factor, $\gamma\in\mathbb{R}$, within the set of nonlinear equations and iteratively reducing it from a value of 1 to 0. Each sub-problem, $\mathcal{H}\left(X,\gamma\right) $ created on the path is solved with the previous homotopy factor iteration solution as an initial condition and is defined as:
\begin{equation}
    \label{eq:homotopy}
    \mathcal{H}\left(X,\gamma\right)=\left(1-\gamma\right)\mathcal{F}\left(X\right)+\gamma E\left(X\right)=0
\end{equation}
where $\gamma \in [0,1]$.

Most importantly, it can be shown that the approach of embedding the homotopy factor to obtain $E\left(X\right)$ influences its solution \cite{allgower2012numerical} and the final solution's trajectory \cite{mcnamara2022two}. \textit{Therefore, designing a homotopy method that utilizes all the available domain-specific knowledge is crucial for tracing a homotopy path toward an optimal and physically meaningful solution to the original problem.}
\begin{lemma}
The equivalent circuit modeling establishes an initial homotopy problem $E(X)$, which is convex.
\end{lemma}

\begin{proof}
Let the generic representation in \eqref{eq:odp_switch} correspond to the desired optimization problem whose solution can be obtained by solving a set of optimality conditions using the homotopy method from \eqref{eq:homotopy}. 

We utilize the equivalent circuit representation to separate and isolate the non-convexities. We build our approach on the recently introduced Incremental Model Building (IMB) homotopy method \cite{pandey2020incremental}. Let the nonconvex terms ($g_{ncvx}\left(X\right)$) of the problem be introduced in a set of separable equality constraints \eqref{eq:odp_kcl_switch}, which is a case for the AC network constraints. In the equivalent-circuit method, the nonconvex terms represent the generators and loads in the system, while the convex components, $g_{cvx}(X)$ model the transmission and other series elements. The IMB formulation embeds the homotopy factor, $\gamma$, into the set of equality constraints by:
\begin{equation}            \label{eq:separate_cvx_equations}g\left(X\right)=g_{cvx}\left(X\right)+(1-\gamma) g_{ncvx}\left(X\right)=0. 
\end{equation}
The initial homotopy problem, $E(X)$, corresponds to the case where $\gamma=1$, and the nonconvex equality constraints are removed (i.e., the generators and loads are shorted). As the objective function and simple bounds remain convex, the entire problem $E(X)$ is convex. 
\end{proof}
After solving for $E\left(X\right)=0$, tracing the homotopy path corresponds to gradually reintroducing the system demand and transmission losses towards the original problem setting. The IMB approach traces the homotopy path of virtually energizing the grid, demonstrating robust convergence to physically operable and optimal solutions \cite{pandey2020incremental}.

While IMB has shown strong convergence characteristics, it is designed to solve ACOPF with fully continuous variables. Direct inclusion of integer variables breaks the IMB as it relies on the smooth properties of the underlying functions. The continuous switch representation of binary decisions in Section \ref{sec:switch_model} now transforms the MINLP into a fully continuous solution space that can utilize powerful homotopy methods such as IMB.

While the continuous switch models ensure a continuous solution space, they further add steep nonlinearities. These switch-like nonlinearities can cause abrupt changes in the solution space from one homotopy step to another. To address this challenge, we augment IMB with two additional homotopy methods, Gmin Stepping and Parallel Conductance Stepping, that use the physical characteristics of the grid to prevent abrupt changes caused by the steep continuous-switch models. %A flowchart describing the complete homotopy process is shown in Fig. \ref{fig:flowchart}.
\iffalse
\begin{figure}
    \centering
    \includegraphics[width=\columnwidth]{Figures/switch_flowchart.png}
    \caption{Methodology for Solving Optimal-Decision Problems Using Continuous Switch Model and Homotopy Methods.}
    \label{fig:flowchart}
\end{figure}
\fi
\subsubsection{Gmin Stepping Augmented Homotopy Method}

%If a device is to be turned on using the switch model during the IMB homotopy path, the diode-like model in \eqref{eq:Igsw_eq} creates an abrupt change that establishes a solution that is far from the one at the previous homotopy step. To address this behavior, we augment the IMB by a circuit simulation homotopy method known as Gmin stepping \cite{pillage1998electronic}. 
The Gmin stepping algorithm, inspired by a circuit-based homotopy method, is used to reduce the abrupt nature of the companion circuit model in Fig. \ref{fig:Igsw_function}. The Gmin-stepping algorithm is a homotopy method that adds a large shunt conductance, $G_{min}$ to the companion circuit model, and effectively shorts the node voltage to zero. Iteratively, the value of $G_{min}$ is reduced until the shunt conductance has a zero value (i.e., the shunt conductance is removed). This traces a path in the solution space that starts from a trivial solution where the node-voltage in the companion model is zero. The shunt conductance is controlled by a homotopy factor, $\gamma_{sw}^{Gmin}$ to its value ($\gamma_{sw}^{Gmin}$), as shown in Fig. \ref{fig:gmin_stepping}.

\begin{figure}
    \centering
\includegraphics[width=0.7\columnwidth]{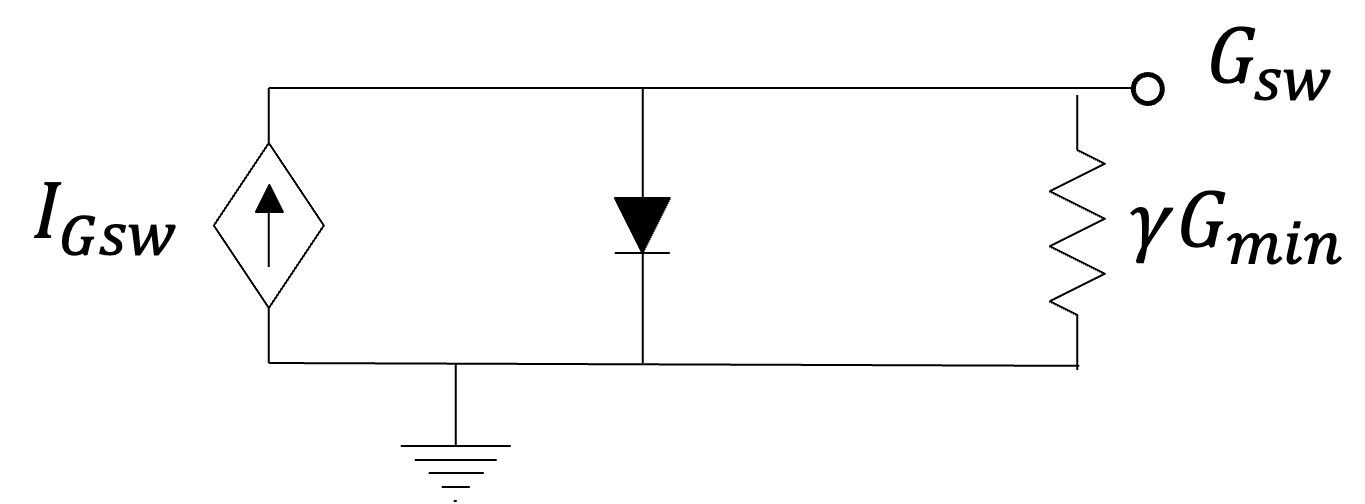}
    \caption{Gmin stepping on diode control circuit with an embedded homotopy factor of $\gamma_{sw}^{Gmin}$.}
    \label{fig:gmin_stepping}
\end{figure}

Since the node-voltage of the companion circuit is a circuit representation of the switch conductance, $G_{sw}$, the large conductance, $G_{min}$, initially forces $G_{sw}\rightarrow0$ and effectively disconnects all switches. This homotopy method works in conjunction with IMB, and the initial problem, $E(X)$, defined by $\gamma=1$, remains convex as the addition of the shorted companion model has a trivial solution of $G_{sw}=0$. 

As we iteratively decrease $\gamma_{sw}$, the Gmin-stepping method traces a solution path that starts to optimally turn switches on as the grid is energized through IMB. However, in the process of reenergizing the grid, $G_{min}$ acts as a buffer to reduce the abruptness of $G_{sw}$ going from 0 to $G_{closed}$. The effect is shown in Fig. \eqref{fig:gmin_stepping_function}.

\begin{figure}
    \centering
    \includegraphics[width=0.56\columnwidth]{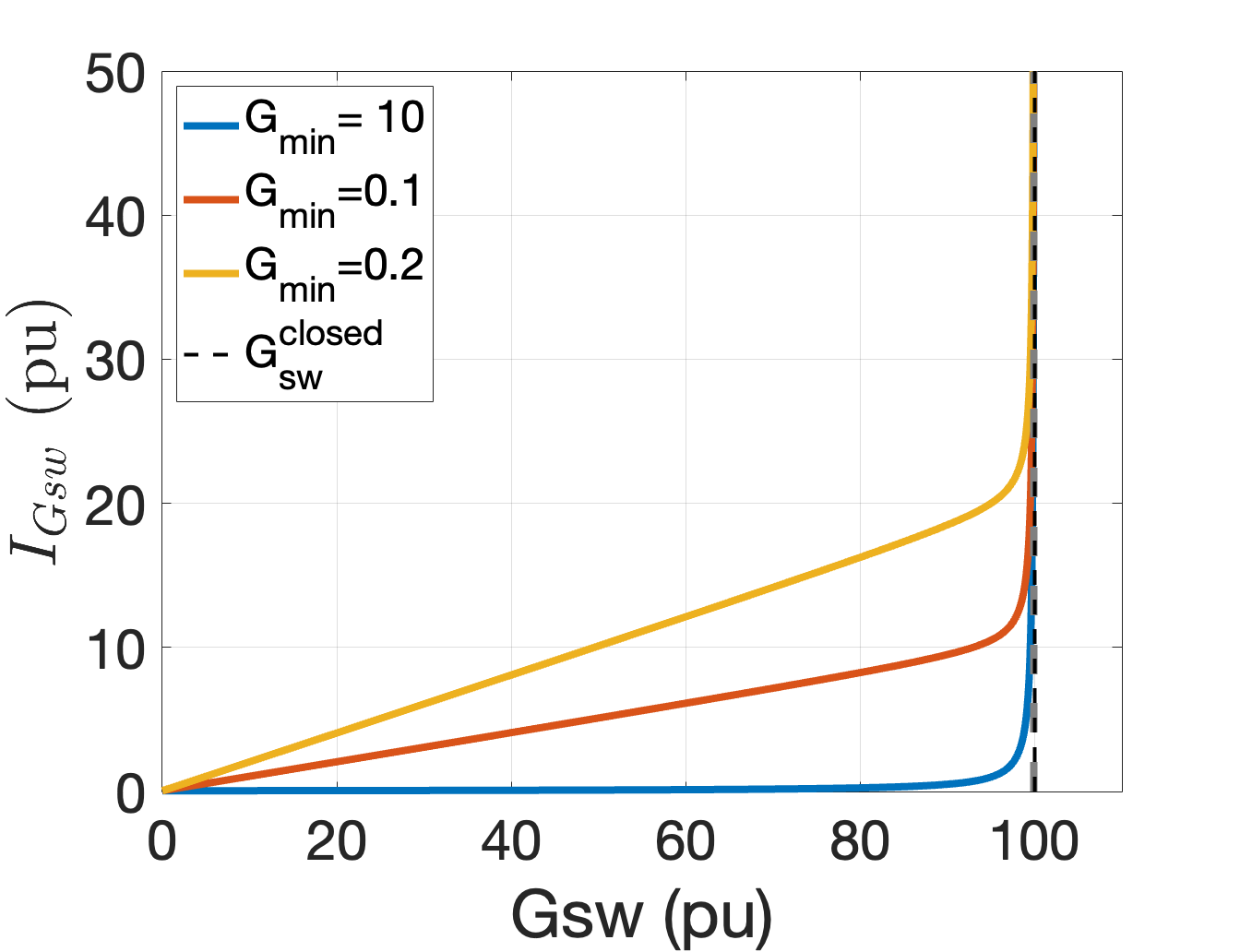}
    \caption{Diode function modified by Gmin stepping. Note that as the homotopy factor is decreased, the current throughout the network ramps up according to IMB, and the effect of the shorted conductance $G_{min}$ reduces. As a result of the increased current throughout the network, the added devices that are necessary to facilitate feasible operation will increase their corresponding switch conductance.}
    \label{fig:gmin_stepping_function}
\end{figure}

Domain knowledge dictates that this approach favors solutions with fewer closed additional devices. We achieve this by beginning the homotopy trajectory with all companion circuits shorted (i.e. $G_{sw}=0$) and, as a result, all added devices are disconnected. Mathematically, the effect of the homotopy dependent shunt conductance ($G_{min}$) can be described by: 
\begin{equation}
    \left(I_{Gsw}-{\gamma_{sw}^{Gmin}G}_{min}G_{sw}\right) \left(G_{sw}-G_{sw}^{closed}\right)+\epsilon=0
    \label{eq:Igsw_eq_gmin_stepping}
\end{equation}

\noindent Analogous to the method in circuit simulation, the shunt conductance is added to ensure diagonal dominance of the overall set of equations and hence positive semi-definiteness \cite{pillage1998electronic} of the diode control equations.

\subsubsection{Parallel Conductance Augmented Homotopy Method}
The Gmin stepping method can prevent abrupt changes in the system states due to sudden changes in the value of the nonlinear conductance, $G_{sw}$ and create a smooth homotopy trajectory. However, using Gmin stepping alone can cause large voltage changes while adding \emph{\textbf{transmission line elements}}. When no current flows through a potential transmission line, $G_{sw}=0$, thereby disconnecting the device from the rest of the grid (assuming the grid is not islanded). Effectively, there is a floating branch (disconnected on either end), so the voltages at from and to nodes will be different. Suppose the homotopy factor were to increase and the switch conductance was to connect the line. In that case, the large voltage difference may cause a large current flow through the line, making it difficult to solve the current homotopy step.

To prevent the voltages across the transmission element from changing abruptly when the switch closes (i.e., $G_{sw}$ increases), we develop an augmented method solved in conjunction with IMB homotopy. We insert a conductance in parallel $G_{sw}^{par}$ to the nonlinear switch conductance $G_{sw}$. This parallel conductance is incrementally stepped down to obtain the solution to the problem at each IMB homotopy step.

As shown in Fig. \ref{fig:conductance stepping}, when no current flows through the transmission element and the switch conductance is 0, the parallel conductance (shown in \textcolor[rgb]{0.85,0,0}{red}) still ensures the voltages across the switch conductance remain the same. As a result, if the switch closes (or the conductance increases) during homotopy, the solutions from subsequent homotopy steps are within a NR quadratic convergence region.
\begin{figure}[!ht]
    \centering
\includegraphics[width=0.6\columnwidth]{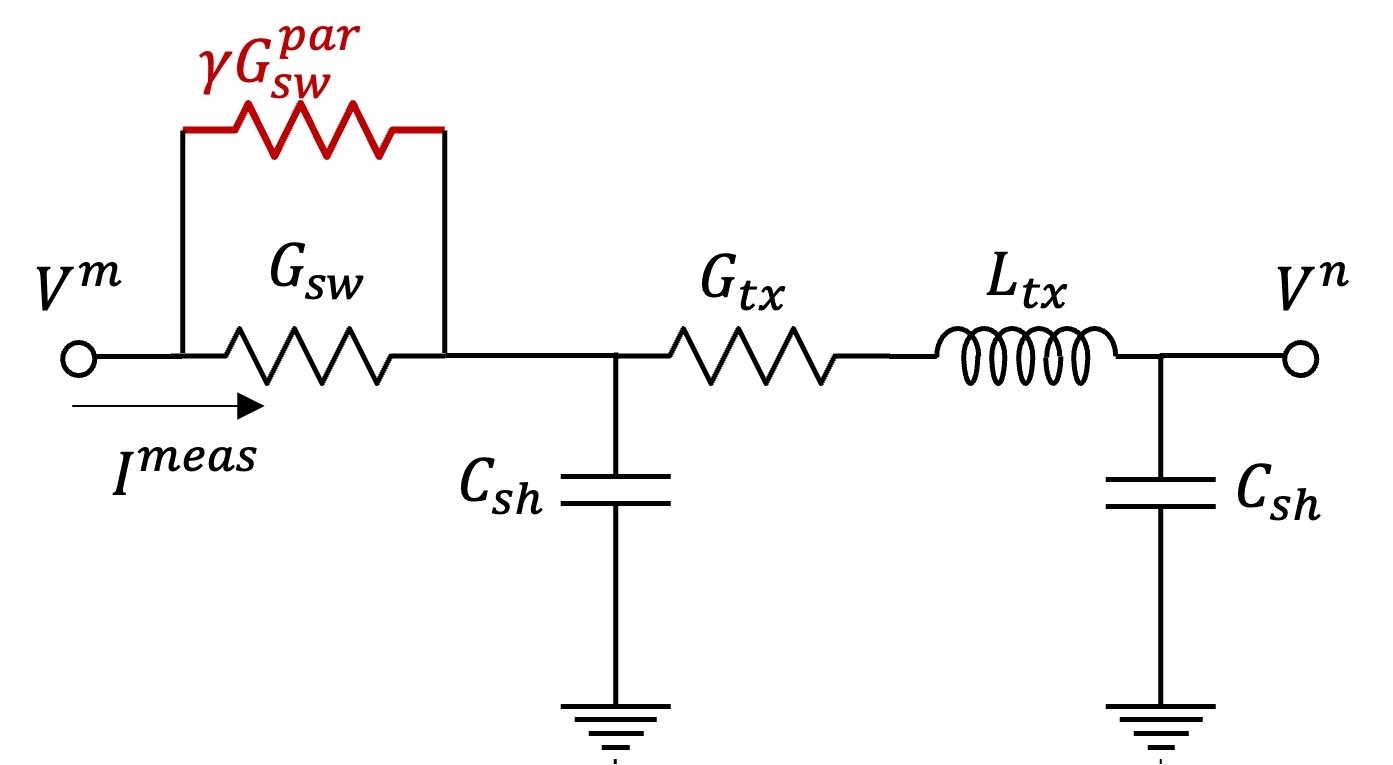}
    \caption{Homotopy method inserts a parallel conductor $\gamma G_{sw}^{par}$ (shown in \textcolor[rgb]{0.85,0,0}{red}).}
    \label{fig:conductance stepping}
\end{figure}
\subsubsection{Convergence of Homotopy Methods}
The two homotopy methods work together to reduce any abrupt change in mimicking the ideal switch behavior. Homotopy methods rely on the following three criteria to guarantee convergence to a final solution at $\gamma=0$ \cite{allgower2012numerical}:\\
	1) The homotopy path, $c\left(\gamma\right)\in H^{-1}(0)$ is smooth \\
	2) The defined homotopy path, $c\left(\gamma\right)\in H^{-1}(0)$ exists \\
	3) The path, $c(\gamma)$ intersects with the final solution at the final homotopy factor, $\gamma=0$

The first criterion is met with the continuous switch model, enabling a fully continuous solution space. The second criterion is met using the slack injection approach in \cite{infeas_marko}. The slack injection approach places homotopy parameter-dependent hypothetical current sources at each bus (represented by a vector $I_{slack}$) to satisfy KCL during the entire homotopy trajectory \cite{pandey2020incremental}. In addition, the objective function is amended from \eqref{eq:odp_obj_switch} to minimize the effect of current slack injection (with a large weight $w_{slack}$) during homotopy, shown in \eqref{eq:odp_total_obj_switch_infeas}. As we multiply all slack injections with the homotopy factor, they disappear at the last step, when the homotopy factor is 0 (i.e., $\gamma = 0$).
 \begin{equation}
     \begin{split}
         \label{eq:odp_total_obj_switch_infeas}
         f\left(X_c,\frac{G_{sw}}{G_{closed}}\right)&=f_{su}\left(\frac{G_{sw}}{G_{closed}}\right)+f_{op}\left(X_c,\frac{G_{sw}}{G_{closed}}\right) \\
         &+{\gamma w}_{slack}|I_{slack}|^2
     \end{split}     
 \end{equation}
 Finally, the convergence guarantee for homotopy methods requires that the path $c(\gamma)$ intersects with the solution at $\gamma=0$. This criterion can be linked to existence theorems in nonlinear analysis \cite{allgower2012numerical}.We can ensure that a path does converge to a solution at $\gamma=0$ by preventing it from extending to infinity. This is achieved using limiting methods that prevent the curve from diverging and extending to infinity. By satisfying the homotopy method criteria, we can ensure that our trajectory converges to a feasible solution for the original problem. 

 \subsection{Newton Raphson Damping}
 Final convergence challenges arise due to the diode-like device in the companion circuit (Fig. \ref{fig:Igsw_function}), which can cause numerical oscillations when there are large update steps in the NR-step. As a result, we dampen NR steps to prevent large changes in the diode-like model. 
Based on diode-limiting work from the field of circuit simulation \cite{pillage1998electronic}, we limit the voltage across the companion circuit, $G_{sw}$ to ensure that $G_{sw} \leq G_{closed}$. However, the steep region near $G_{sw}\approx G_{closed}$ poses a challenge during NR as any large NR step can cause $G_{sw}$ to exceed $G_{closed}$ and cause numerical overflow problems. In this approach, we damp the value of $G_{sw}^{k+1}$ at the $k+1$ iteration of NR using a damping factor, $\tau_{Gsw}$, to ensure the switch conductance value does not exceed its upper limit ($G_{closed}$). Therefore, the next iterate of $G_{sw}$ is determined by the damping factor and the update step, $\Delta G_{sw}^k$ (determined from the NR step).

\begin{equation}
    G_{sw}^{k+1}=G_{sw}^k+\tau_{Gsw}\Delta G_{sw}^k
\end{equation}

The damping factor $\tau_{Gsw}$calculated in \eqref{eq:tau_gsw}, along with a scaling factor, $\eta\le1$, ensures $G_{sw}$ does not reach its limit exactly (a value of $\eta=0.95$ works well).

\begin{equation}
    \label{eq:tau_gsw}
    \tau_{Gsw}=min(1,\eta\frac{\left(G_{closed}-G_{sw}^k\right)}{\Delta G_{sw}^k})
\end{equation}

To further improve convergence robustness with limiting methods, we note that within an epsilon region where $G_{sw}$ approaches $G_{closed}$ in \eqref{eq:Igsw_eq}, the switch model is current controlled, i.e., any change in current, $I_{Gsw}$ will change the value of $G_{sw}$ by an $\epsilon$ amount. As a result, for a point above a critical value of $G_{sw}^{crit}$, $I_{Gsw}$ is used to establish $G_{sw}$, as shown below.

\begin{subequations}
\begin{equation}
    if \ G_{sw}^k\geq G_{sw}^{crit}:\ \ \ \ \ \ G_{sw}^{k+1}=-\frac{\epsilon}{I_{mag}^{k+1}}+G_{closed}    
\end{equation}
\begin{equation}
    if\ G_{sw}^k<G_{sw}^{crit}:\ \ \ \ \ \ G_{sw}^{k+1}=G_{sw}^k+\tau_{Gsw}\Delta G_{sw}^k
\end{equation}
\end{subequations}
The critical junction value of $G_{sw}^{crit}$ is determined by the point of maximum radius of curvature of \eqref{eq:Igsw_eq} that is a result of the relaxation, $\epsilon$:
\begin{equation}
    G_{sw}^{crit}=G_{closed}-\sqrt\epsilon
\end{equation}

\section{Case Studies}
\label{sec:results}
The GISMO framework can solve several planning and operation ODP problems, including optimally adding generation and switching transmission elements. We validate the method's optimality on a small testcase and then demonstrate the scalability by studying large synthetic US Eastern Interconnection-sized networks \cite{xu2017creation} with high renewable penetration. The results highlight that GISMO is: 1) robust to the choice of initial conditions, 2) can provide an AC-feasible solution unlike many other mixed-integer solvers, 3) generalizable to solve generation and transmission expansion simultaneously, and 4) scalable to solve large-scale systems with up to $\sim$70k nodes. Importantly, we demonstrate the efficacy of GISMO by comparing it against industry-standard approaches including commercial mixed-integer software such as Gurobi and KNITRO, as well as other relaxation methods solved using a widely-used MATLAB optimization solver: fmincon.

\subsection{Comparison of 14-Bus Network Expansion}
To validate the GISMO approach, we study the reconstruction of a small network, where we identify an optimal set of generator and transmission lines to switch on to supply a system demand. In this small-scale 14-bus IEEE testcase, all branches (17 transmission lines and 2 transformers) and all 5 generators are considered potential grid devices to add to the grid, each with a switching cost. The switching cost of a transmission line is chosen to be 2000 units, and the start-up cost of the generation is 100,000 units. %The generator linear operational cost is provided in Table \ref{tab:14_bus_gen_cost}. 

To determine the optimal set of devices to supply the system load, a continuous switch model is added in series to all generators and branches. 
\iffalse
Specifically, the equality constraints in \eqref{eq:odp_kcl_switch} are modified to the equations in \eqref{eq:14_bus_Igen}-\eqref{eq:14_bus_sh}. The KCL constraints for generator currents, $I_{gen}$, are modified to include the series switch conductance $G_{sw}$ between the generator node voltage, $V_{gen}^i$ and the bus voltage, $V_{bus}^i$ in \eqref{eq:14_bus_Igen}. Similarly, the equality constraint for each transmission line (with current $I_{tx}^i$) is modified in \eqref{eq:14_bus_tx} to include the switch conductance between the voltages at the end of the transmission line ($V_{tx}^{from}$) and the bus $V_{bus}^{from}$). Lastly, the switch conductance is added in series to the shunt capacitor in \eqref{eq:14_bus_sh}.
\begin{subequations}
\begin{equation}
    \label{eq:14_bus_Igen}
    I_{gen}^i+G_{sw}^i\left(V_{gen}^i-V_{bus}^i\right)=0
\end{equation}
\begin{equation}
    \label{eq:14_bus_tx}
    I_{tx}^i+G_{sw}^i\left(V_{tx}^{from}-V_{bus}^{from}\right)=0
\end{equation}
\begin{equation}
    \label{eq:14_bus_sh}
    I_{sh}^i+G_{sw}^i\left(V_{sh}^i-V_{bus}^i\right)=0
\end{equation}
\end{subequations}
\fi
The relaxed problem is then solved using GISMO. We determine that a set of 11 transmission lines, 2 transformers, and 3 generators (Fig. \ref{fig:14_bus_expansion}) will optimally supply the system demand while satisfying AC network constraints. The total cost of the optimized network in Fig. \ref{fig:14_bus_expansion} is given in Table \ref{tab:14_bus_total_cost}.
\begin{figure}
    \centering\includegraphics[width=0.65\columnwidth]{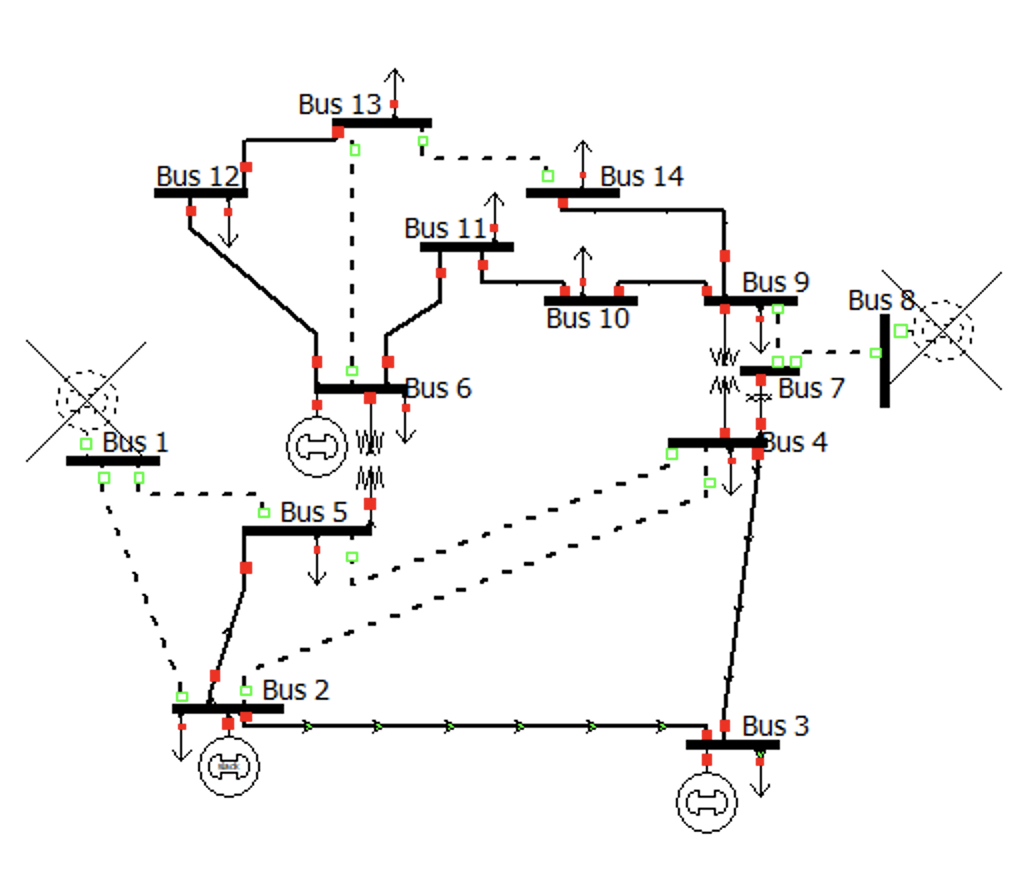}
    \caption{Reconstruction of a 14 bus system with all transmission lines and generators as potential grid expansion devices.}
    \label{fig:14_bus_expansion}
\end{figure}

\begin{table}
\caption{Linear Cost for Generators in 14 Bus Network\label{tab:14_bus_gen_cost}}
\centering
\begin{tabular}{|c|c|c|c|c|c|}
\hline
Gen. Bus & 1 & 2& 3&6&8 \\ 
\hline
Cost/MW & 1100 &1000&900&800&700 \\
\hline
\end{tabular}
\end{table}

\begin{table}
\caption{Linear Cost for Generators in 14 Bus Network\label{tab:14_bus_total_cost}}
\centering
\begin{tabular}{|c|c|c|}
\hline
Installation Cost & Operation Cost & Total Cost \\ 
\hline
39,00 & 3,392 & 42,392 \\
\hline
\end{tabular}
\end{table}

To validate the optimality of the expanded network solution via brute force, we would have to simulate the operation of $2^{24}$ scenarios (24 is the number of devices). Since permuting through all possible network configurations is intractable, we select a subset of those configurations as represented by those with at least 3 generators and 10 transmission lines. Any fewer transmission lines would leave the network islanded and fewer generators would satisfy the system load. From the subset of network configurations, we identify 17 feasible configurations. The total cost for the feasible network configurations is shown in Fig. \ref{cost_14_bus_expansion}. The network configuration with the lowest cost, denoted by the orange marker in Fig. \ref{cost_14_bus_expansion}, is identical to the optimal configuration we obtained in Fig. \ref{fig:14_bus_expansion}. This validates that GISMO can obtain the most optimal set of generators and transmission lines (without any prior information) necessary for the moot reconstruction of a  network.

\begin{figure}
    \centering
    \includegraphics[width=0.55\columnwidth]{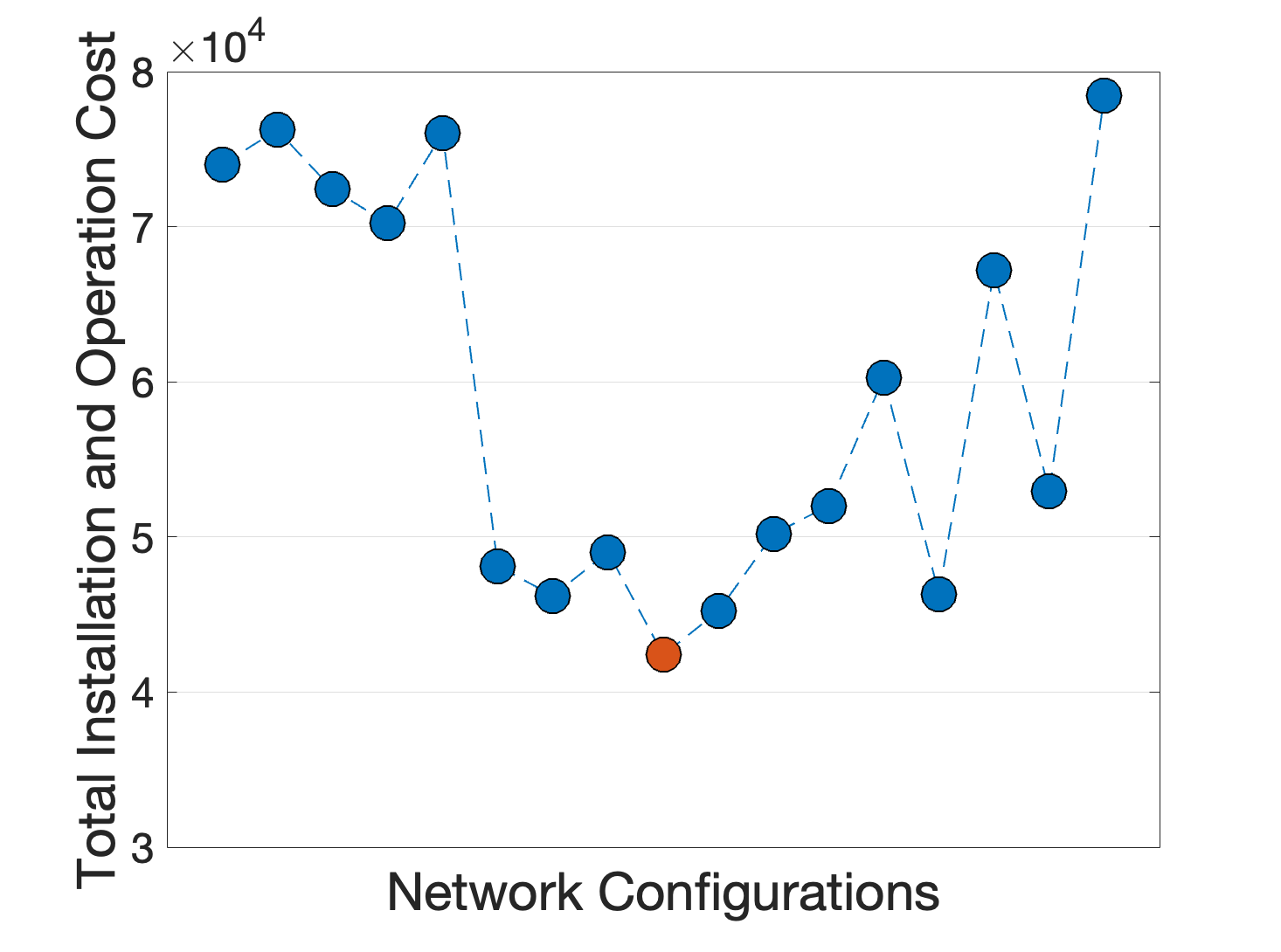}
    \caption{Total Operational and Installation cost of 14 bus network configurations with the lowest shown in orange.}
    \label{cost_14_bus_expansion}
\end{figure}

\subsubsection{Comparison with KNITRO}
We compare the solution robustness of GISMO against KNITRO \cite{byrd2006k}; a commercial mixed-integer nonlinear optimization tool that uses branch-and-bound methods to solve the ODP. As shown in Table \ref{tab:14_bus_knitro}, KNITRO, with relaxation induced neighborhood search method, cannot converge for the AC-constrained problem in \eqref{eq:odp} when starting from an initial condition defined by the optimal power flow solution of the original IEEE-14 bus test case. In contrast, when initialized from the GISMO solution, KNITRO converges to provide identical solution. 
\begin{table}
\caption{Results for AC-Constrained 14-Bus Reconstruction Using KNITRO Initialized With Various Initialization \label{tab:14_bus_knitro}}
\centering
\begin{tabular}{|c|c|c|}
\hline
  & Init. with 14-bus solution & Initialized with GISMO solution \\ 
\hline
\#iter & N/A (not converged) & 3 \\
\hline
\end{tabular}
\end{table}

\subsubsection{Comparison with DC-based Expansion on Gurobi \cite{gurobi2021gurobi}}
To avoid the difficulty of nonlinear AC constraints, many researchers \cite{zhang2012transmission,danna2005exploring,ma2020unit,huang2016optimal,huang2017branch,haghighat2018bilevel,gao2022internally,ke2015modified,lopez2021optimal,chen2016improving,nasri2015network,miao2015capacitor} relax the problem by using approximate linear DC constraints (essentially converting the problem into MILP). However, this can lead to an infeasible real-world setting. We study the expansion of the same 14-bus system but with DC network constraints to demonstrate the drawback of ignoring AC constraints. Using a commercial MILP solver, Gurobi \cite{gurobi2021gurobi} with default settings and identical operational and switching costs, the solver converges to a configuration shown in Fig. \ref{fig:14_bus_DC_solution}. However, when verifying the solution with AC power flow, we recognize that the network is infeasible, as the transmission line from bus 2 to 3 becomes overloaded and voltages throughout the network are not within operational bounds. This highlights the need to consider AC constraints within the expansion problem implicitly.
\begin{figure}
    \centering
    \includegraphics[width=0.65\columnwidth]{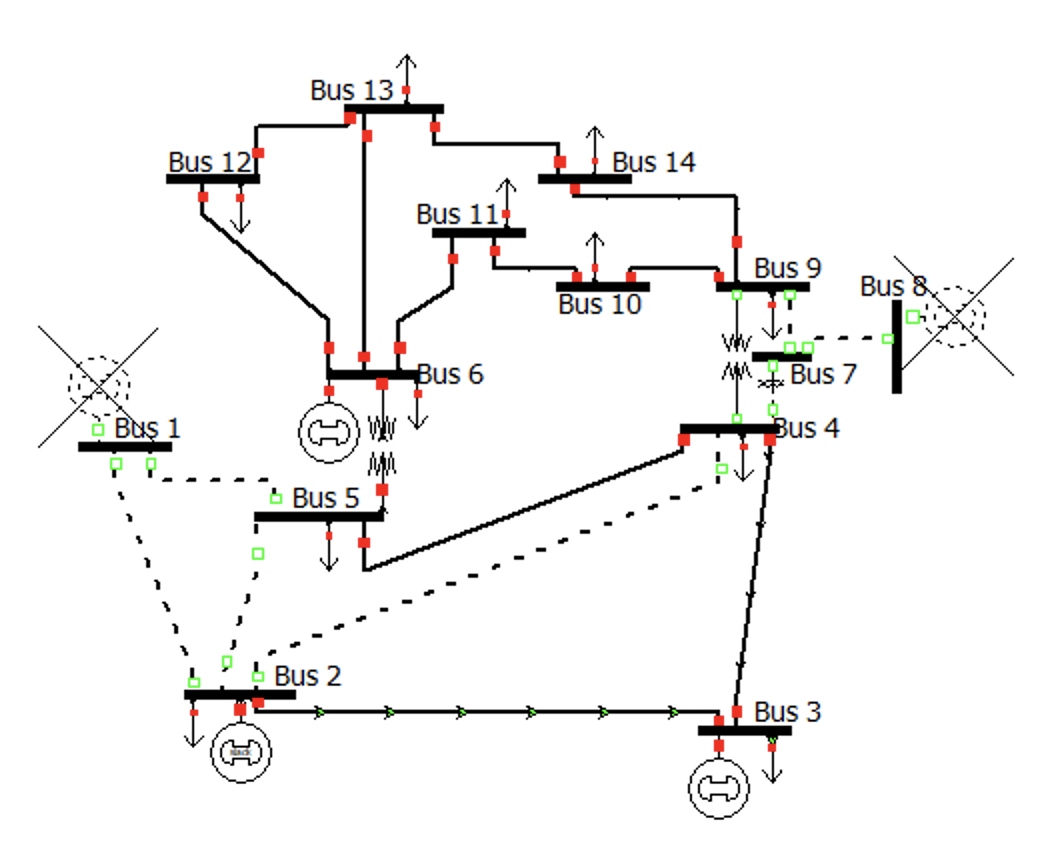}
    \caption{Infeasible 14-bus expansion network solution from Gurobi DC-constrained optimization.}
    \label{fig:14_bus_DC_solution}
\end{figure}

\subsubsection{Comparison with Other Relaxation Methods}
Previous methods have used relaxation methods to solve MINLP by relaxing the integrality constraints. A common binary-relaxation method is to modify the objective with the addition of \cite{zhang2012transmission}:
\begin{equation}
    f\left(X_c,\ X_b\right)+w_b(X_b\left(1-X_b\right))
\end{equation}
and remove the integrality constraint in \eqref{eq:odp_integrality}. The addition of $w_b(X_b\left(1-X_b\right))$ heavily penalizes (with a weight of $w_b$) $X_b$ to be anything other than 0 or 1 \cite{zhang2012transmission}. The solution to the relaxed problem using the quadratic binary approximation is then snapped to the nearest binary value. 

To demonstrate the need for physics-inspired homotopy and damping heuristics alongside the relaxation, we apply the quadratic relaxation to the 14-bus reconstruction problem. Using fmincon \cite{byrd2000trust}, we initialize the system from a flat start with $X_b=0$, and $w_b=1000$. Due to the quadratic relaxation, the solver finds the optimal solution where most values of $X_b$ are around 0.4. When snapped to the nearest integer (in this case, 0), many devices are turned on, and the system is infeasible. However, when the system is initialized with the solution from the continuous solution, an identical solution to GISMO is achieved in 2 iterations.
\begin{table}
\caption{Results for AC-Constrained 14-Bus Reconstruction Using Quadratic Binary Relaxation \label{tab:14_bus_quadratic}}
\centering
\begin{tabular}{|c|c|c|}
\hline
  & Init. with flat start,$X_b=0$ & Initialized with GISMO solution \\ 
\hline
\#iter & N/A (not converged) & 2 \\
\hline
\end{tabular}
\end{table}

\subsection{Scaling to Synthetic Eastern Interconnection System}
To demonstrate the scalability of GISMO, we study three cases of a modified synthetic US Eastern Interconnection testcase \cite{xu2017creation} (70,000 buses): I) renewable expansion, ii) limited unit commitment, and iii) transmission line switching. The network is available at \cite{testcases}.

\subsubsection{Renewable Expansion of Synthetic Eastern Interconnection}
While this study is a proof of concept to demonstrate scalability, we follow a similar methodology to a recent renewable integration and expansion study by Midcontinent Independent System Operator (MISO) \cite{bakke2019renewable}, where they studied the effects of integrating 40\% renewable sources. We modify the original ACTIVSg70k testcase by decommissioning coal-powered plants and replacing them with renewable installations. The modified testcase follows a similar trend as the future expansion study by increasing renewable penetration and the system load. This introduces low voltage issues into the modified case, as shown in red on the left in Fig. \ref{fig:EI_voltage}. Current strategies for expansion that use DC-based methods cannot observe and mitigate the low voltage effects shown in Fig. \ref{fig:EI_voltage}. This highlights the need for requiring AC constraints in any expansion study.
\begin{figure}
    \centering
    \includegraphics[width=0.8\columnwidth]{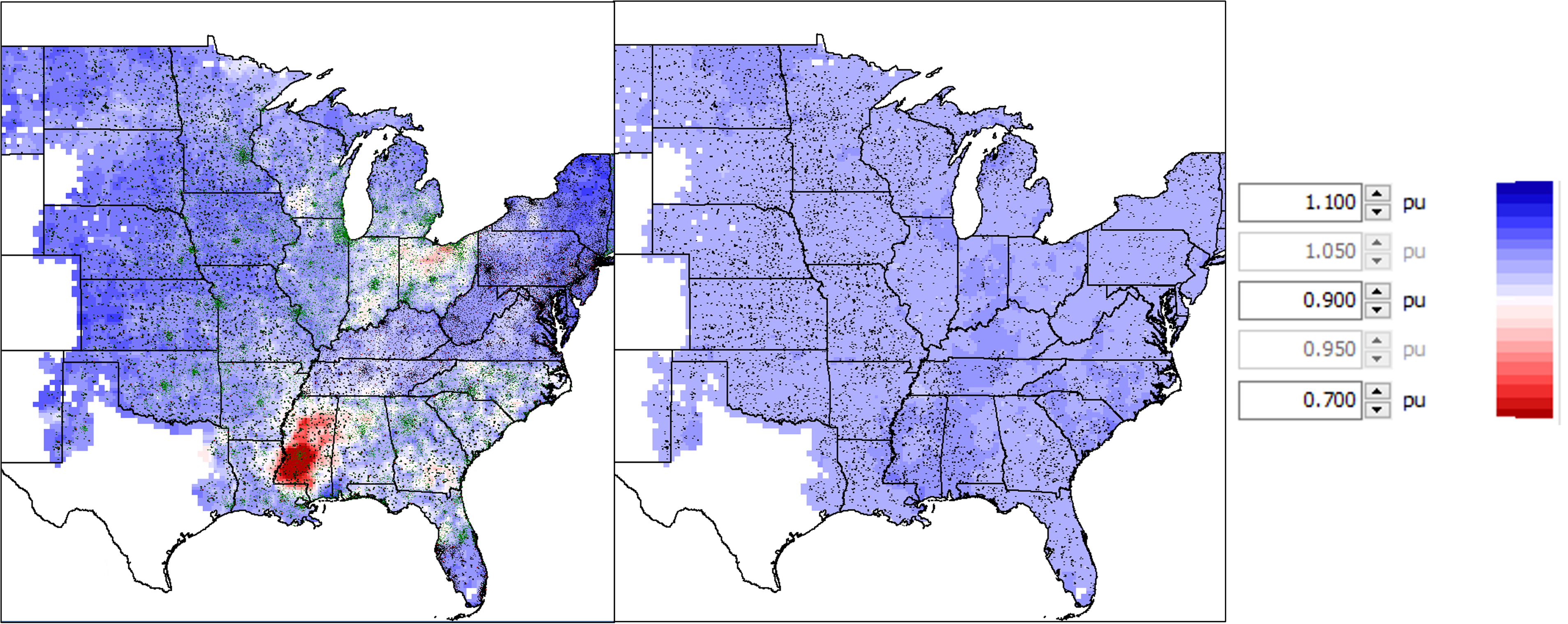}
    \caption{Voltage profile of modified ACTIV70k with 35\% renewable penetration (left) and with transmission expansion (right).}
    \label{fig:EI_voltage}
\end{figure}

We look to reinforce the low-voltage synthetic grid using GISMO (which uses AC constraints) by commissioning a set of potential devices, including 100 transmission lines, 100 generators, and 100 shunt banks. Each potential device is connected to the grid with a series continuous switch, which is optimized using GISMO. Based on the solution from the proposed approach, we identify 4 transmission lines, and two shunts that guarantee an AC-operational dispatch with improved voltage stability, as shown in Fig. \ref{fig:EI_voltage}. The total operational and commissioning cost is in Table V and the added devices are shown in Fig. \ref{fig:EI_expansion}. 
\begin{figure}
    \centering
    \includegraphics[width=0.5\columnwidth]{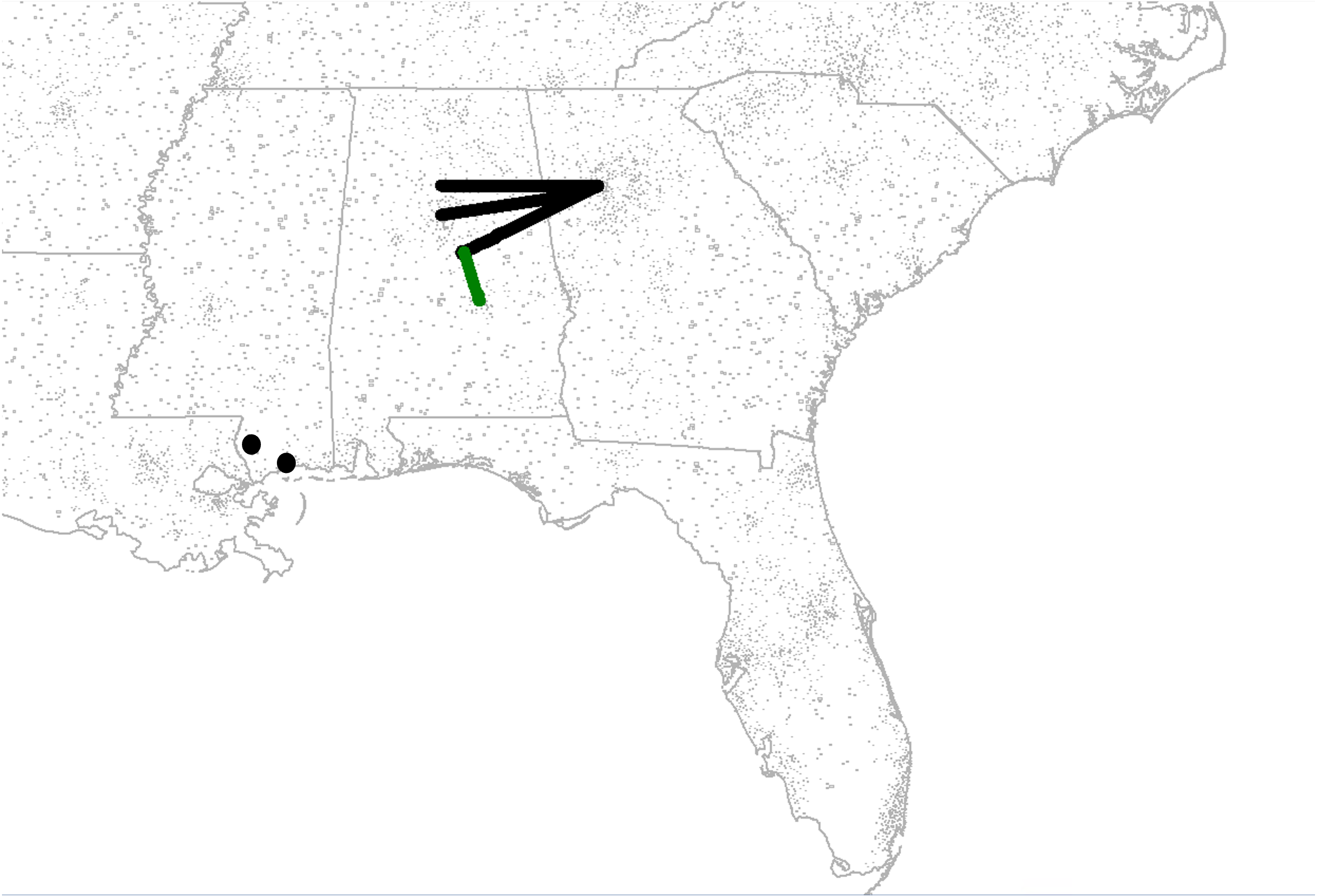}
    \caption{Optimal ACTIV70k system with four new transmission elements and two shunt as indicated by black lines and dots respectively}
    \label{fig:EI_expansion}
\end{figure}

\begin{table}
\caption{Total Cost of Installation and Operation of ACTIVSg70k \label{tab:EI_cost}}
\centering
\begin{tabular}{|c|c|c|}
\hline
  Installation Cost & Operation Cost & Total Cost \\ 
\hline
78,000 & 3,081,902 & 3,159,902 \\
\hline
\end{tabular}
\end{table}

To demonstrate that such large-scale networks cannot be solved without applying homotopy and limiting heuristics, we attempted to solve the expansion of the ACTIVSg70k testcase using the continuous switch models in fmincon \cite{byrd2000trust}. fmincon is initialized using flat start conditions (1.0 pu voltage and 0.0 pu angles) and default parameters. As shown in Table V, fmincon with the continuous switch model cannot converge when initialized from a flat start, while GISMO converges with flat start initial conditions. To validate that our solution was correct, we also initialize fmincon with the solution from GISMO and it converges in 2 iterations to the exact solution.

\begin{table}[t]
\caption{Iterations for AC-Constrained Expansion of Synthetic Eastern Interconnection Using Fmincon and GISMO \label{tab:EI_iterations}}
\centering
\begin{tabular}
{>{\raggedright}p{0.04\linewidth}|
 >{\raggedright}p{0.32\linewidth}|
 >{\raggedright}p{0.32\linewidth}|
 p{0.1\linewidth}}
\hline\noalign{\smallskip}
 &\textbf{fmincon $\leftarrow$ flat start} 
 &\textbf{fmincon $\leftarrow$ GISMO solution}
 &\textbf{GISMO} \\ 
\hline
Iter\#& Did not converge & 2 & 287 \\
\hline
\end{tabular}
\footnotesize{1. fmincon is formulated with continuous switch formulation in \eqref{eq:odp_switch}}
\end{table}

\subsubsection{Limited Unit Commitment of Synthetic EI}
GISMO is applicable for solving unit commitment for large case systems. In the following experiment, we solve a limited unit commitment problem, defined in the ARPA-E GO Competition \cite{challenge}, in which “fast-start” generators can start-up to support the system load for the modified Eastern Interconnection testcase for a single time-window.

The testcase is initially infeasible as insufficient generation is switched on to supply the network load. GISMO adds a continuous switch model in series with the fast-start generators to add an optimal set of generation that ensures AC-feasibility and minimal start-up and operational cost. The result is that 5 additional generators are turned on, as shown in Fig. 14 (left).  

\subsubsection{Transmission Line Switching for EI}
We also apply GISMO for optimal transmission line switching to relieve congestion. The ACTIV70k testcase is modified by decreasing the line limits to induce congestion in the lines. GISMO considers a set of 100 transmission lines that can be switched on by placing a continuous switch model in series with each potential transmission line. By optimizing for the lowest switching cost, GISMO identifies 5 transmission lines, shown in Fig. 14 (right), to relieve congestion and ensure the grid is AC-feasible. %All data for all case studies in this section are available in \cite{testcases}.
\begin{figure}
    \centering
    \includegraphics[width=\columnwidth]{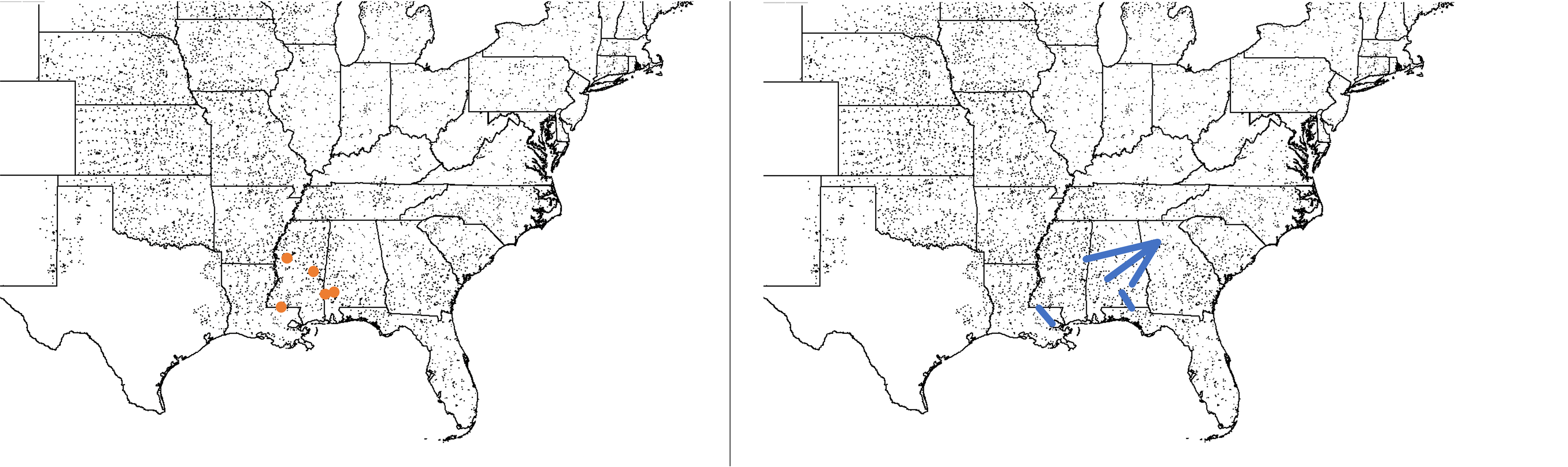}
    \caption{Additional devices turned on for unit commitment (left) and transmission line switching (right) in modified ACTIV70k testcase}
    \label{fig:EI_uc_switching}
\end{figure}

\section{Conclusion}
Generally posed as a mixed-integer nonlinear problem, the optimal-decision problem in power systems represents a growing set of planning and operation analyses. We introduce a new equivalent-circuit framework, GISMO, to solve the optimal-decision problem by modeling the binary decisions as a continuous switch model. The continuous switch provides physical insights to develop strong heuristics in the form of homotopy methods and Newton-Raphson dampening that provide scalable and robust convergence. The methodology is shown to be more robust to an operational solution when compared to existing MINLP solvers and other relaxation methods. Additionally, the scalability of the methodology is demonstrated by optimally selecting elements to stabilize a high renewable penetration study of a synthetic Eastern Interconnection case with over 70,000 buses. The continuous switch methodology provides a general approach for solving many vital and upcoming analyses for the grid.

%\section*{Acknowledgments}
%This should be a simple paragraph before the References to thank those individuals and institutions who have supported your work on this article.

%\bibliographystyle{IEEEtran}
\printbibliography

\end{document}